\documentclass[11pt]{article}
\usepackage{fullpage}
\usepackage{amsmath}
\usepackage{amssymb}
\usepackage{amsthm}
\usepackage{comment}
\usepackage{paralist}

\newcommand{\ens}[1]{\ensuremath{#1}}

\newcommand{\bigoh}[1]{\ens{\mathcal{O}(#1)}}
\newcommand{\bigom}[1]{\ens{\Omega(#1)}}

\newcommand{\varbigoh}[1]{\ens{\mathcal{O}\left(#1\right)}}
\newcommand{\varbigom}[1]{\ens{\Omega\left(#1\right)}}
\newcommand{\varthet}[1]{\ens{\Theta\left(#1\right)}}

\newcommand{\ozc}[1]{\textcolor{Red}{$\rightarrow$#1$\leftarrow$}}
\renewcommand{\ozc}[1]{}

\newtheorem{theorem}{Theorem}
\newtheorem{lemma}[theorem]{Lemma}
\newtheorem{definition}[theorem]{Definition}

\usepackage{multibib}
\newcites{shortbib}{References}
\newcites{fullbib}{References}

\usepackage{graphicx}
\usepackage[usenames,dvipsnames]{color}
\usepackage{rotating}
\usepackage{algpseudocode}
\usepackage{algorithm}
\usepackage{url}

\allowdisplaybreaks

\newtheorem{obs}[theorem]{Observation}

\newcommand{\superscript}[1]{\ensuremath{^{\textrm{#1}}}}

\newcommand{\johnc}[1]{{\sc \textcolor{red}{(xiii) #1}}}
\newcommand{\johnd}[1]{{\sc \textcolor{red}{(xiv) #1}}}
\renewcommand{\johnc}[1]{}
\renewcommand{\johnd}[1]{}

\newwrite\notationfile
\openout\notationfile=\jobname.notation

\newcommand{\jnc}[3]{\write\notationfile{$\mathrm{\string\string\string#1}$ & \ensuremath{#2} & #3 \string\\}\newcommand{#1}{\ensuremath{#2}}}

\newcommand{\opIns}{\ensuremath{\text{\sc Insert}}}
\newcommand{\opEm}{\ensuremath{\text{\sc Extract-Min}}}
\newcommand{\opDc}{\ensuremath{\text{\sc Decrease-Key}}}
\newcommand{\opMg}{\ensuremath{\text{\sc Meld}}}

\newcommand{\eIns}{\ensuremath{\text{\sc Evolve-Insert}}}
\newcommand{\eDc}{\ensuremath{\text{\sc Evolve-Decrease-Key}}}
\newcommand{\eDmr}{\ensuremath{\text{\sc Evolve-Designated-Minimum-Root}}}
\newcommand{\eEm}{\ensuremath{\text{\sc Evolve-Extract-Min}}}
\newcommand{\eBs}{\ensuremath{\text{\sc Evolve-Big-Small}}}
\newcommand{\ePerm}{\ensuremath{\text{\sc Evolve-Permute}}}

\newcommand{\tent}[1]{\overline{#1}}

\newcommand{\subop}[2]{{\sc #1(\ensuremath{#2})}}

\newcommand{\cst}{\frac{n  }{\rcf {} \log n}} 
\newcommand{\pp}{++}

\jnc{\Ids}{\ensuremath{\Xi}}{?}
\jnc{\Distinct}{\ensuremath{\xi}}{?}
\jnc{\opa}{\ensuremath{a}}{Description}
\jnc{\opb}{\ensuremath{b}}{Description}
\jnc{\opc}{\ensuremath{c}}{Description}
\jnc{\opA}{\ensuremath{A}}{Description}
\jnc{\opB}{\ensuremath{B}}{Description}
\jnc{\opC}{\ensuremath{C}}{Description}
\jnc{\subops}{\eta}{Description}
\jnc{\pointers}{\rho}{Description}
\jnc{\pointer}{p}{Description}
\jnc{\state}{s}{Description}
\jnc{\rank}{r}{Description}
\jnc{\alg}{\mathcal{A}}{The algorithm with $o(\log \log n)$-time decrease key which is being assumed to exist.}
\jnc{\cem}{\alpha}{The (maximum of) the constants hidden in the big-O of the $O(\log n)$-time \opEm\ and \opIns\ operations. }
\jnc{\rce}{j(n)}{Description}
\jnc{\rcf}{f(n)}{Description}
\jnc{\rcg}{g(n)}{Description}
\jnc{\rch}{h}{Maximum fraction of the time spent doing fix-up}
\jnc{\rcl}{\ell}{Amortized runtime is $\ell \log n$ per round}
\jnc{\fdc}{dc(n)}{Cost of decrease-key}
\jnc{\vio}{v}{Description}
\jnc{\akey}{x}{A key value or node}
\jnc{\theset}{\mathcal{S}}{Set of items stored by the data structure}
\jnc{\acem}{\ensuremath{e}}{...}
\jnc{\dcc}{dc}{Number of \opDc\ performed in a permutation evolution in round $i$}
\jnc{\vs}{v}{Size of the violation sequence of the \opEm\ in round $i$}
\jnc{\csmall}{z}{At least $z$ of the rounds are small}
\jnc{\rootsize}{m}{The size of a particular root}
\jnc{\noninc}{N}{Nonincremental sibling set}
\jnc{\bigroot}{r}{Location of largest root}
\jnc{\rcd}{{\ensuremath{j(n)}}}{Used in rank definition}
\jnc{\round}{\circ}{Round}
\jnc{\umc}{k}{Number of unmarked children}
\jnc{\Vio}{V}{Violation set: those nodes that become marked in a extract-min evolution}
\jnc{\Classification}{C}{Classification in a big-small evolution}
\jnc{\Evo}{\Psi}{A sequence of evolutions}
\jnc{\evo}{\psi}{One evolution in a sequence of evolutions}
\jnc{\runb}{R(B)}{Runtime to execute operation sequence $B$}
\jnc{\rcm}{m}{Constant involved in size of number of decrease-key operations in small round}
\jnc{\ozwin}{\ozpexp}{Constant for the range of things that are efficient}
\jnc{\ozwinpo}{(\ozpexp+1)}{\ozwin{} plus 1}
\jnc{\nozwinmo}{\ozpexp-1}{Negative \ozwin{} minus 1}
\jnc{\ozexp}{1+ \log (32\cem+32\rcm)}{Constant for the exponent of \rcd{}}
\jnc{\ozpexp}{c}{Placeholder for \ozexp}
\jnc{\ftop}{\rcd^{\ozpexp}}{Shortcut}
\jnc{\aitsd}{\text{AI-TSD}}{AI-TSD}
\closeout\notationfile

\title{A Tight Lower Bound for Decrease-Key in the Pure Heap Model}

\author{John Iacono\thanks{Research supported by NSF Grant CCF-1018370.} \and \"Ozg\"ur \"Ozkan}

\date{{\sc New York University}\\Brooklyn, New York, USA}

\makeatother

\begin{document}

\maketitle

\begin{abstract}
	
We improve the lower bound on the amortized cost of the decrease-key operation in the pure heap model 
and show that any pure-heap-model heap (that has a \bigoh{\log n} amortized-time extract-min operation) must spend \bigom{\log\log n} amortized time on the decrease-key operation.  
Our result shows that sort heaps as well as pure-heap variants of numerous other heaps have asymptotically optimal decrease-key operations in the pure heap model. 
In addition, our improved lower bound matches the lower bound of Fredman
[J. ACM 46(4):473-501 (1999)]
for pairing heaps
[M.L. Fredman, R. Sedgewick, D.D. Sleator, and R.E. Tarjan. Algorithmica 1(1):111-129 (1986)]
 and surpasses it for pure-heap variants of numerous other heaps with augmented data such as pointer rank-pairing heaps. 
\end{abstract}

\section{Introduction} 
\label{sec:introduction}
A \emph{priority queue} is a fundamental data structure that supports the following operations to maintain a totally ordered set\footnote{The key values can be in any form so long as they are totally ordered and an $O(1)$-time comparison function is provided. This is more permissive than saying they they are comparison-based. We do not restrict algorithms from doing things like making decisions based on individual bits of a key. Some priority queues, including Fibonacci heaps and our sort heaps, also efficiently support the \opMg\ operation where two priority queues are combined into one.
} \theset:
$\pointer =\ $\opIns$(\akey)$: inserts the key $\akey$ into $\theset$ and returns a pointer $p$ used to perform future \opDc\ operations on this key,
$\akey =\ $\opEm$()$: removes the minimum item in $\theset$ from $\theset$ and returns it,
and \opDc$(\pointer, \Delta \akey)$: reduces the key value pointed to by $\pointer$ by some non-negative amount $\Delta \akey$.

The fast execution of \opDc\ is vital to the runtime of several algorithms, most notably Dijkstra's algorithm \cite{springerlink:10.1007/BF01386390}. The constant-amortized-time implementation of the \opDc\ operation is the defining feature of the \emph{Fibonacci heap} \cite{DBLP:journals/jacm/FredmanT87}. 

However, Fibonacci heaps are not pointer model structures in the standard textbook presentation \cite{clrs} for the following reason: each node is augmented with a $\log \log n$-bit integer (in the range $1\ldots \log n$), and in the implementation of \opEm\ it is required to separate a number of nodes, call it $k$, into groups of nodes with like number. This is done classically using bucket sort in time \bigoh{k+\log n}. 
Bucket sort using indirect addressing to access each bucket is distinctly not allowed in the pointer model; realizing this, there is a note in the original paper \cite{DBLP:journals/jacm/FredmanT87}  showing how by adding another pointer from every node to a node representing its bucket, bucket sort can be simulated at no additional asymptotic cost. However, these pointers are non-heap pointers to nodes with huge in-degree which store no keys and thus violate the tree requirement of the pure heap model, as well as the spirit of what we usually think of as a heap. 

There have been several alternatives to Fibonacci heaps presented with the same amortized runtimes which are claimed to be simpler than the original~\cite{DBLP:journals/talg/KaplanT08,DBLP:journals/dmaa/Elmasry10,DBLP:journals/siamcomp/HaeuplerST11}.

All of these heaps also use indirect addressing or non-heap pointers. The rank-pairing heap~\cite{DBLP:journals/siamcomp/HaeuplerST11} has the cleanest implementation of all of them, implementing \opDc\ by just disconnecting the node from it parents and not employing anything more complicated like the cascading cuts of Fibonacci heaps.
We can easily modify Fibonacci heaps and the aforementioned alternatives to only use heap pointers by simply 
using \bigoh{\log \log n} time to determine which of the \bigoh{\log n} buckets each of the $k$ keys are in. We call such a variant of rank-pairing heaps a \emph{pointer rank-pairing} heap. However, this alteration increases the time of \opDc\ in  Fibonacci heaps and their variants to \bigoh{\log \log n}. 

In \cite{DBLP:journals/algorithmica/FredmanSST86}, a new heap called the \emph{pairing heap} was introduced. The pairing heap is a self-adjusting heap, whose design and basic analysis closely follows that of splay trees \cite{DBLP:journals/jacm/SleatorT85}. They are much simpler in design than Fibonacci heaps, and they perform well in practice. 
It was conjectured at the time of their original presentation that pairing heaps had the same \bigoh{\log n} amortized time \opEm, and \bigoh{1} amortized time \opIns\ and \opDc\ as Fibonacci heaps, however, at their inception only a \bigoh{\log n} amortized bound was proven for all three operations. 
Iacono has shown a working-set like runtime bound \cite{DBLP:conf/swat/Iacono00}, and also that \opIns\ does in fact have \bigoh{1} amortized time \cite{DBLP:journals/corr/abs-1110-4428,DBLP:conf/swat/Iacono00}. 
However in \cite{DBLP:journals/jacm/Fredman99}, Fredman refuted the conjectured constant-amortized-time \opDc\ in pairing heaps by proving that pairing heaps have a lower bound of \bigom{\log \log n} on the \opDc\ operation. The result he proved was actually more general: he created a model of heaps which includes both pairing heaps and Fibonacci heaps, and produced a tradeoff between the number of bits of data augmented and a lower bound on the runtime of \opDc. Pairing heaps have no augmented bits, and were shown to have a \bigom{\log \log n} amortized lower bound on \opDc\ while in his model a \bigoh{1} \opDc\ requires \bigom{\log \log n} bits of augmented information per node, which is the number of bits of augmented information used by Fibonacci heaps and variants
 (a number called rank, which is an integer with logarithmic range is stored in each node). 
 More recently, Pettie has shown a upper bound of \bigoh{2^{2 \sqrt{\log \log n}}}  for \opDc\ in pairing heaps \cite{DBLP:conf/focs/Pettie05}. It remains open where in $\bigom{\log \log n} \ldots \bigoh{2^{2 \sqrt{\log \log n}}}$ the cost of \opDc\ in a pairing heap lies.
Elmasry has shown that a simple variant of pairing heaps has \bigoh{\log \log n} amortized time \opDc\ operation \cite{DBLP:conf/esa/Elmasry10,DBLP:conf/soda/Elmasry09}. However, 
this variant is not in Fredman's model and thus the \bigom{\log \log n} lower bound does not apply.

Recently, a new model of pointer model heaps called the \textit{pure heap} model was introduced~\cite{DBLP:conf/icalp/IaconoO14}, along with a variant of Elmasry's heap in the model called the sort heap which supports the \opDc\ operation in \bigoh{\log \log n} amortized time. 
A pure heap is a pointer-based forest of rooted trees, each node holding one key, that obeys the heap property (the key of the source of every pointer is smaller than the key of the destination); the nodes may be augmented, and all operations must be valid in the pointer model; heap pointers are only removed if one of the nodes is removed or if a \opDc\ is performed on the node that a heap pointer points to. 
This model is both simple and captures the spirit of many heaps  like the pairing heap, and is meant to be a clean definition analogous to that of the well-established binary search tree (BST) model \cite{DBLP:journals/siamcomp/Wilber89}.  

It was shown in~\cite{DBLP:conf/icalp/IaconoO14} that any heap in the pure heap model that has \bigoh{\log n} amortized-time \opEm{} and \opIns{} operations must spend \varbigom{\frac{\log\log n}{\log\log\log n}} amortized time on the \opDc{} operation, \emph{no matter how many bits of data each node is augmented with}. 
We improve this result by showing that any heap in the pure heap model that has \bigoh{\log n} amortized-time \opEm{} and \opIns{} operations must in fact spend \bigom{\log \log n} amortized time on the \opDc{} operation.
This bound is asymptotically tight as it matches the upper bounds for numerous heaps including pointer rank-pairing heaps and variants, and sort heaps. 
It also matches (or surpasses in the case of pointer rank-pairing heaps) the lower bounds achieved in Fredman's model.
In this view, Fibonacci heaps are a typical RAM-model structure that squeeze out a $\log \log n$ factor over the best structure in a natural pointer-based model by beating the sorting bound using bit tricks (of which bucket sort is a very primitive example). 
Thus, Fibonacci heaps have fast \opDc\ is not (only) because of the augmented bits as Fredman's result suggests, but rather because they depart from the pure heap model.

We review the pure heap model in Section~\ref{sec:the_pure_heap_model} and present the improved lower bound in Section~\ref{sec:lower_bound}.


\section{The Pure Heap Model} 
\label{sec:the_pure_heap_model}

Here we define the \emph{pure heap} model, and how priority queue operations on data structures in this model are executed in it.
The pure heap model requires that at the end of every operation, the data structure is an ordered forest of general heaps. 
Each node is associated with a key $\akey \in \theset$. We will use $x$ both to refer to a key and the node in the heap containing the key. Inside each node is stored the key value, pointers to the parent, leftmost child and right sibling of the node, along with other possible augmented information.
The \emph{structure} of the heap is the shape of the forest, without regard to the contents of the nodes. The \emph{position} of a node is its position in the forest of heaps relative to the right. The idea is that location is invariant under adding new siblings to the left.

An algorithm in the pure heap model implements the priority queue operations as follows:
\opIns\ operations are executed at unit cost by adding the new item as a new leftmost heap. 
\opDc\ is executed at unit cost by disconnecting from its parent the node containing the key to be decreased (if it is not a root), decreasing the key, and then placing it as the leftmost root in the forest of heaps. 
An \opEm\ operation is performed by first executing a sequence of pointer-based suboperations which are fully described below.
After executing the suboperations, the forest is required to be monoarboral (i.e.~have only one heap). 
Thus, the root of this single tree has as its key the minimum key in $\theset$. This node is then removed, the key value is returned, and its children become the new roots of the forest. The cost is the number of suboperations performed. 

Note that some data structures are not presented exactly in the framework as described above but can be easily put into this mode by being lazy. For example, in a pairing heap, the normal presentations of \opIns\ and \opDc\ cause an immediate pairing with the single existing root. However, such pairings can easily be deferred until the next \opEm, thus putting the resulting structure in our pure heap framework. 
For a more elaborate example of transforming a heap into one based on pairings, see \cite{DBLP:conf/wae/Fredman99}.

\subsection{Executing an \opEm} \label{suboperations}

To execute an \opEm, the minimum must be determined. 
In an \opEm, the forest of heaps must be combined into a single heap using an operation called \emph{pairing}, which takes two roots and attaches the root with larger key value as the leftmost child of the root with smaller key value.
(Note that while the pairing operation brings to mind pairing heaps, it is the fundamental building block of many heaps. Even the \emph{skew heap} \cite{DBLP:journals/siamcomp/SleatorT86}, which seems at first glance to not use anything that looks like the pairing operation, can be shown in all instances to be able to be transformed into a pairing-based structure at no decrease in cost \cite{DBLP:conf/wae/Fredman99}.)
In the execution of the \opEm\ operation, the use of some constant number $\pointers$ of pointers $\pointer_1, \pointer_2, \ldots \pointer_\rho$ is allowed. They are all initially set to the leftmost root. 
\begin{enumerate}

\item \subop{HasParent}{i}: \label{so:chpar} Return if the node pointed to by $\pointer_i$ has a parent.

\item \subop{HasLeftSibling}{i}: Return if the node pointed to by $\pointer_i$ has a left sibling.

\item \subop{HasRightSibling}{i}: Return if the node pointed to by $\pointer_i$ has a right sibling.

\item \subop{HasChildren}{i}: \label{so:chchild} Return if the node pointed to by $\pointer_i$ has any children.

\item \subop{Compare}{i,j}:\label{so:compare} Return if the key value in the node pointed to by $p_i$ is less than or equal to the key value in the node pointed to by $p_j$.

\item \subop{Pair}{i,j}:\label{so:pair} Perform a \emph{pairing} on two pointers $\pointer_i$ and $\pointer_j$ where the tree that $p_j$ points to is attached to $\pointer_i$ as its leftmost subtree. It is a required precondition of this suboperation that both $p_i$ and $p_j$ point to roots, and that this was verified by the \textsc{HasParent} suboperation, and that the key value in the node pointed to by $p_i$ is smaller than the key value in the node pointed to by $p_j$, and that is was verified by the \textsc{Compare} suboperation.

\item \subop{Unpair}{i}: \label{so:unpair} Disconnect the left child of the node $p_i$ points to and make it the left sibling of $p_i$. To use this suboperation, it is required that $p_i$ has children, and crucially that $p_i$ and its left child were paired with a \subop{Pair}{\cdot,\cdot} suboperation performed in the same \opEm\ as this one. Without this restriction, proof of the Lemma~\ref{lem:monotonerank} below would not hold.

\item \subop{Set}{i,j}: Set pointer $\pointer_i$ to point to the same node as $\pointer_j$.

\item \subop{Swap}{i,j}: Swap the heaps pointed to by $\pointer_i$ and $\pointer_j$.

\item \subop{MoveToParent}{i}: Move a pointer $\pointer_i$ to the parent of the node currently pointed to.
It is a precondition of this operation that the node that $\pointer_i$ points to has a parent, and that this was verified by the \textsc{HasParent} operation.

\item \subop{MoveToLeftmostChild}{i}: Move the pointer $\pointer_i$ to the leftmost child.
It is a precondition of this operation that the node that $\pointer_i$ points to has children, and that this was verified by the \textsc{HasChildren} operation.

\item \subop{MoveToRightSibling}{i}: Move the pointer $\pointer_i$ to the sibling to the right.
It is a precondition of this operation that the node that $\pointer_i$ points to has a right sibling, and that this was verified by the \textsc{HasRightSibling} operation.

\item \subop{MoveToLeftSibling}{i}: \label{so:mvleft} Move a pointer $\pointer_i$ to the sibling to the left
It is a required precondition of this operation that the node that $\pointer_i$ points to has a left sibling, and that this was verified by the \textsc{HasLeftSibling} operation.

\item \subop{End}{}: \label{so:end} Marks the end of the suboperation sequence for a particular \opEm.
It is a required precondition of this operation that the forest of heaps contains only one heap, and that this was verified through the use of the \subop{HasParent}{i}, \subop{HasLeftSibling}{i}, and \subop{HasRightSibling}{i} on a pointer $p_i$ that points to the unique root.

\end{enumerate}

Operations \ref{so:chpar}-\ref{so:chchild} return a boolean; the remainder have no return value.

The total number of suboperations, including the parameters, is defined to be $\subops$. Observe that $\subops=\Theta(\pointers^2)$, which is $\Theta(1)$ since $\rho$ is a constant. 
A sequence of suboperations is a \emph{valid} implementation of the \opEm\ operation if all the preconditions of each suboperation are met and the last suboperation is an \subop{end}{}.
In the pure heap model, the only thing that differentiates between different algorithms is in the choice of the suboperations to execute \opEm\ operations. In these operations it is the role of the particular \emph{heap algorithm} to specify which suboperations should be performed for each \opEm. There are no restrictions as to how an algorithm determines the suboperation sequence for each \opEm\ other than the suboperation sequence must be valid; the algorithm need not restrict its actions to information in the nodes visited in that operation.


\section{Lower Bound} 
\label{sec:lower_bound}

\begin{theorem}[Main Theorem] \label{thm:main}
In the pure heap model with a constant number of pointers, if \opEm{} and \opIns{} have an amortized cost of \bigoh{\log n}, then \opDc{} has an amortized cost of \bigom{\log\log n}.
\end{theorem}

The proof will follow by contradiction. Assume that there is a pure heap model algorithm \alg{} where \opEm{} and \opIns{} have an amortized cost of at most $\cem \log n$, for some constant \cem{}, and \opDc{} has an amortized cost of at most \fdc, for some $\fdc=o(\log\log n)$.
The existence of the algorithm  \alg{},  the constant \cem{} and the function \fdc{} will be assumed in the definitions and lemmas that follow. A sufficiently large $n$ is also assumed.

\subsection{Preliminaries} 
\label{sub:preliminaries}

The proof is at its core an adversary argument. Such arguments look at what the algorithm has done and then decide what operations to do next in order to guarantee a high runtime. But, our argument is not straightforward as it works on sets of sequences of operations rather than a single operation sequence. There is a hierarchy of things we manipulate in our argument:

\begin{description}

\item[Suboperation.] The suboperations of \S\ref{suboperations} are the very basic unit-cost primitives that can be used to implement \opEm{}, the only operation that does not have constant actual cost. It is at this level that definitions have been made to enforce pointer model limitations.
\item[Operation.] We use \emph{operation} to refer to a priority queue operation. In this proof, the adversary will only use \opIns, \opDc, and \opEm.
\item[Sequence.] Operations are combined to form sequences of operations. 
\item[Set of operation sequences.] Our adversary does not just work with a single operation sequence but rather with sets of operation sequences. These sets are defined to have certain invariants on the heaps that result from running the sequences which bound the size of the sets of operation sequences under consideration.
\item[Evolution.] We use the word \emph{evolution} to refer to a function the adversary uses to take a set of operation sequences, and modify it. The modifications performed are to append operations to sequences, remove sequences, and to create more sequences by taking a single sequence and appending different operations to the end.
\item[Rounds.] Our evolutions are structured into \emph{rounds}.

\end{description}

The proof will start with a set containing a single operation sequence, and then perform rounds of evolutions on this set; the exact choice of evolutions to perform will depend on how the algorithm executes the sequences of operations in the set. The evolutions in a round are structured in such a way that most rounds increase the size of the set of operations. After sufficiently many rounds, an upper bound on the maximum size of the set of operation sequences will be exceeded, thus giving a contradiction.

In the remainder of Section~\ref{sub:preliminaries}, we provide numerous definitions and present a number of notions needed for the proof. 
In Section~\ref{sub:outline_of_the_proof}, we will give a more detailed outline of the proof. 
In Sections~\ref{sub:evolutions} and \ref{sub:rounds}, we describe our evolutions and rounds as defined above, respectively. 
Finally, in Sections~\ref{sub:upper_bound_on_time} and \ref{sub:obtaining_a_contradiction} we derive a contradiction to prove our result.

\subsubsection{Rank of a node} 
\label{ssub:rank_marked_unmarked_nodes}

As in many previous works on heaps and trees, the notion of the \emph{rank} of a node in the heap is vital. The rank of a node is meant to be a rough proxy for the logarithm of the size of the subtree of the node. 

While the basic analysis of pairing heaps and splay trees \cite{DBLP:journals/algorithmica/FredmanSST86,DBLP:journals/jacm/SleatorT85} use exactly this as the rank, the definition of rank here is more delicate and is an extension of that used in \cite{DBLP:journals/jacm/Fredman99}. As in \cite{DBLP:journals/jacm/Fredman99}, rank here is always a nonnegative integer. In previous definitions of rank, the value typically depended only on the current structure of the heap (One exception to this has been in order to get better bounds on \opIns, nodes are treated differently for potential purposes depending on whether or not they will ever be deleted. See \cite{DBLP:journals/cacm/StaskoV87,DBLP:journals/corr/abs-1110-4428,DBLP:conf/swat/Iacono00}). Here, however, the definition is more nuanced in that for the purposes of the analysis only, nodes are classified into \emph{marked} and \emph{unmarked} categories based on the history of the structure, and these marks, along with the current structure of the heap, are used to compute the rank of each node.

For ease of presentation, the rank of a node is defined in terms of the function $\rce=2\fdc+1$.

The general idea is to have the rank of a node be maintained so it is the negation of the key value stored in the node. (Ranks will be non-negative, and we will only give nodes non-positive integer key values; these can be assumed to be perturbed arbitrarily to give a total ordering of key values). A node's rank can increase as the result of a pairing, and a node's value can decrease as the result of a \opDc. It is thus our goal to perform a \opDc\ on a node which has had its rank increase to restore it to the negation of its rank. During the time between when a rank increase occurs in a node and the time the \opDc\ is performed, we refer to the node as \emph{marked}. 

Call the \emph{unmarked subtree} of a node to be the subtree of a node if all marked nodes were detached from their parents; the \emph{unmarked structure} of the heap is the structure of the unmarked subtrees of the roots. The rank of a node at a given time will be defined as a function of the structure of its unmarked subtree.

 We emphasize that the notion of marking a node is for the purposes of the analysis only, such marks need not be stored.

The following assumes a particular heap structure and marking, as the rank of a node is always defined with respect to the structure of the heap after executing a sequence of heap operations.

Let $\akey$ be the node we wish to compute the rank of. Let $\umc$ denote the number of unmarked children of $\akey$, and let $y_1,y_2,\ldots y_{\umc}$ denote these children numbered right-to-left (i.e., in the order which they became children of $\akey$).

Let $\tau_i(x)$ be a subtree of $x$ consisting of $x$ connected to only the subtrees induced by $y_1, y_2, \ldots y_i$. We will define the function $\rank_i(x)$ as a function of $\tau_i(x)$. The rank of a  node, $\rank(x)$ is $\rank_k(x)$.

Each node $y_i$ is said to be \emph{efficiently linked} to its parent $x$ if and only if
$ \rank_{i-1}(x)-\ozwin \leq \rank(y_i) \leq \rank_{i-1}(x)$.
The case of $ \rank(y_i) > \rank_{i-1}(x)$ will never occur, as pairings will only happen among unmarked nodes, where the rank perfectly matches the negation of the key value.

We will have the property that $\rank_i(x)$ is either $\rank_{i-1}(x)$ or $\rank_{i-1}(x)+1$; in the latter case, $y_i$ is called \emph{incremental}.

Given a node $y_i$, let $j$ be defined to be the index of the first incremental node in the sequence $\langle y_{i-1}, y_{i-2}, \ldots\rangle$; $j$ is defined to be 0 if there is no such incremental node.
The set $\noninc(y_i)$ is defined to be $\{ y_k | j<k \leq i\}$; that is, $y_i$ and the maximal set of its non-incremental siblings to the right.

Let \ozpexp{} be a constant to be determined later in Section~\ref{sub:obtaining_a_contradiction}. 
Given these preliminaries, we can now give the full definition of the rank of a node:

$$
\rank_i(x)=
\left\{
	\begin{array}{ll}
		0  & \mbox{if } i = 0 \\ \\
		r_{i-1}(x)+1 & \mbox{\parbox{4in}{(Efficient case) $y_i$ is efficiently linked and is the $\rce $th efficient element of $\noninc(y_i)$
		\\%
		---or--- \\
		(Default case) $|\noninc(y_i)|=\ftop {}$}}\\ \\
		r_{i-1}(x) & \mbox{otherwise}
	\end{array}
\right.
$$

While the rank and mark are interrelated, there is no circularity in their definitions---whether a node is marked depends on its rank and key value and the rank of a node is a function of the ranks and marks of its children.

\begin{obs}
\label{obs:structrank}
Given two nodes $x$ and $y$ with different ranks, the unmarked structure of their induced subtrees must be different.
\end{obs}

This follows directly from the fact that the rank of a node is a function of its induced unmarked subtree.

\begin{definition}
Set $\rcf=\frac{\ftop}{2\log\rcd}$ and $\rcg=\frac{\rcd}{2\log\rcd}$. 
\end{definition}

These two functions are defined for convenience. Note that $\rcf 
= o(\log n )$ easily since $\fdc=o(\log \log n)$ and \ozpexp{} is a constant.

Next, we present a key lemma that related the number of unmarked subtree size to the number of efficiently linked children. 

\begin{lemma}\label{lem:effcount}
Suppose a root with unmarked subtree of size $\rootsize$ has $\leq {\rcf  \log \rootsize}$ unmarked children. Then it has $\geq  {\rcg \log \rootsize}$ efficiently linked children. 
\end{lemma}

\begin{proof}
We first prove that the size of an unmarked subtree rooted at a node of rank $k$ is at most $\rce ^k$.
Let $s_k$ be the size of the maximum unmarked heap of rank $k$. Such a heap can be created from a maximal unmarked heap of rank $k-1$, which has been paired to $\rce-1$ maximal unmarked heaps of rank $k-1$ and $\ftop-(\rce-1)$ maximal unmarked heaps of rank $k-\nozwinmo$. Thus:
$$s_k \leq (\rce-1) s_{k-1}+(\ftop-(\rcd-1)) s_{k-\nozwinmo}.$$
By induction, 
$$s_k \leq (\rce-1) \rce ^{k-1}+(\ftop -\rcd+1)\rce ^{k-\nozwinmo} 
\leq
\rce ^{k} .$$ 

This implies that if there are $m$ nodes in a node $x$'s unmarked induced subtree, the rank of $x$ is at least $\log_{\rce} {m} $.

Suppose node $v$ has rank $\geq k$ and at most $ \frac{k}{2}\ftop$ unmarked children. Then, $v$ has at least $k\rcd/2$ efficiently linked unmarked children, each having rank $<k$.
To see this consider how the rank of a node is computed. 
If node $v$ has less than $k\rcd/2$ efficiently linked children, then the rank increases due to the efficient case must be less than $k/2$ by the definition of rank. This implies that at least $k/2$ rank increases will be due to the default case. This implies that there are at least $\frac{k}{2}\ftop$ unmarked children, a contradiction.

Consider a root with unmarked subtree size $\rootsize$
and thus rank at least $\log_{\rce } {\rootsize}$. Suppose it has $\leq \frac{\log_{\rce }{\rootsize}}{2}\ftop = \rcf \log m $ unmarked children. Then it has $\geq \frac{1}{2} \rcd\log_{\rce } m = \rcg \log m$ efficiently linked unmarked children by the argument in the previous paragraph.
\end{proof}

\subsubsection{Monotonic operation sequences} 
\label{ssub:monotonic_operation_sequences}
Call the \emph{designated minimum root} the next node to be removed in an \opEm.

\begin{definition}[Monotonic operation sequence] \label{def:monotonic}

Define a \emph{monotonic operation sequence} to be one where \opDc\ operations are only performed on roots, children of the designated minimum root, or marked nodes.

\end{definition}

All of the sequences of operations we define will be monotonic. 
Recall the invariant that the rank of a node is the negation of its key value. 
The following lemma tells us that when the rank of a node changes, in order to recover the invariant, we can call \opDc{} with a positive parameter $\Delta x$. 
The only case when this is not feasible is when the node is the designated minimum root. However, in this case the node will be removed by an extract-min operation. 

\begin{lemma}[Monotonic sequences and rank]
\label{lem:monotonerank} 
In a monotonic operation sequence, the rank of a node never decreases, from the time it is inserted until the time it becomes the designated minimum root.
\end{lemma}
\begin{proof}
In a monotonic operation sequence, for any node $x$ with descendent $y$ (at the beginning of an operation) where all nodes on the path from $x$'s child down to and including  $y$ are unmarked, $y$ will remain in the same location in $x$'s subtree (at the beginning of subsequent operations) until $x$ becomes the designated minimum root. 
Thus, since the rank of a node is a function of the structure of the unmarked subtree of the node which does not change in a monotonic operation sequence until the node becomes designated minimum root, the rank of a node never decreases until it becomes the designated minimum root. 
\end{proof}


\subsubsection{Augmented suboperations} 
\label{ssub:augmented_suboperations}

We augment suboperation~\ref{so:pair}, the \subop{Pair}\cdot\ operation, to return whether or not the rank was incremented as a result of the pairing. This augmentation does not give any more power to the pure heap model, since the exact ranks of all nodes is a function of the suboperation sequence and is thus known to the algorithm already.
We use this augmentation to create a finer notion of what constitutes a distinct sequence of suboperations. 
In particular, we will use the following fact:

\begin{lemma} \label{lem:subdistinct}
Suppose $s_i$ and $s_j$ are two structurally distinct states of the data structure. Suppose a single valid sequence of suboperations implementing an \opEm\ is performed on both, and the outcomes of all augmented suboperations that have return values are identical in both structures. Then, the position of all nodes who have had their ranks changed is identical in both.
\end{lemma}

\begin{proof}
Observe that as a valid sequence of suboperations implementing \opEm{} is performed on both $s_i$ and $s_j$, the labels of the pointers they perform each suboperation on are determined by a right to left order as defined in Section~\ref{sec:the_pure_heap_model} and are thus equivalent. In addition, the only time a node can have its rank change is when something is paired to it, and this is now explicitly part of the return value of suboperation~\ref{so:pair}.
\end{proof}


\subsubsection{Distinctness} 
\label{ssub:distinctness}

Let $\opB=\langle \opb_1, \opb_2, \ldots \rangle$   be a sequence of priority queue operations.
Let $\opA(\opb_i)=\langle \opa^i_1, \opa^i_2, \ldots \rangle$ be the sequence of augmented suboperations and their return values used by algorithm \alg\ to execute operation $\opb_i$ if $\opb_i$ is an \opEm ; if it is not $\opA(\opb_i)$ is defined to be the empty sequence. $\opA(\opB)$ is the concatenation of $\opA(\opb_1), \opA(\opb_2), \ldots$.

We call two sequences of priority queue operations $B$ and $B'$ \emph{algorithmically indistinguishable} if $A(B) =A(B')$, else they are \emph{algorithmically distinct}.

Let $\state_{\opB}(i)$ be the structure of the heap after running sequence $\langle \opb_1\ldots \opb_i\rangle$; the \emph{terminal structure} of $B$ is $\state_{\opB}(|\opB|)$ which we denote as $\state_{\opB}$. Recall that by structure, we mean the raw shape of the heap without regard to the data in each node, but including which nodes are marked. Two sequences $\opB$ and $\opB'$ are \emph{terminal-structure indistinguishable} if $\state_B = \state_{B'}$, else they are \emph{terminal-structure distinct}.
Given a set of mutually algorithmically indistinguishable and terminal-structure distinct (AI-TSD) sequences of heap operations $\Ids$, the \emph{distinctness} of the set, $\Distinct(\Ids)$ is defined to be $\log |\Ids|$.

Note that having two sequences which are algorithmically indistinguishable does not imply anything about them being terminal-structure indistinguishable. For example, it may be possible to add a \opDc\ to a sequence, changing the terminal structure, while the sequence of suboperations performed to execute the sequence remains unchanged. 
(Recall that suboperations only occur during \opEm\ operations).

A critical observation needed at the end of the proof is that the number of terminal-structure distinct sequences is function of $n$:

\begin{lemma} \label{lem:maxdistinct}
The maximum distinctness of any set $\Ids$ of terminal-structure distinct sequences, all of which have terminal structures of size $n$, is $\Distinct(\Ids)=\bigoh{n}$.
\end{lemma}

\begin{proof}
The number of different shapes of a rooted ordered forest with $n$ nodes is $C_n$, the $n$\superscript{th} Catalan number. The number of different ways to mark some nodes in a forest with $n$ nodes is $2^n$. Since $C_n \leq 4^n$, the maximum number of distinct structures is at most $\log (2^n4^n) = \bigoh{n}$. 
\end{proof}



\subsection{Outline of the proof} 
\label{sub:outline_of_the_proof}

In order to prove the result, we will show that assuming $\fdc = o(\log\log n)$ leads to a contradiction, namely that of Lemma~\ref{lem:maxdistinct}.
We obtain this contradiction using the help of evolutions which will be described in detail in Section~\ref{sub:evolutions}.

Evolutions take as input \aitsd{} sets of sequences of heap operations and modify them by appending operations to all sequences or by removing some sequences. 
Evolutions are performed in rounds with
the goal increasing the distinctness of the set of sequences of heap operations until a contradiction to Lemma~\ref{lem:maxdistinct} is obtained.

The only evolution that increases distinctness is the permutation evolution (Section~\ref{ssub:permutation_evolution}) which given an \aitsd{} set of sequences selects a set of nodes in the terminal structure of each sequence (let us call this set the permutable set), and appends sequences of \opDc{} operations applied to every permutation of these nodes to increase the size of the set. 
Translating this increase in the size of the set to a contradiction of Lemma~\ref{lem:maxdistinct} comes with a few challenges. 

First, the resulting set must be terminal-structure distinct. This is achieved by ensuring the permutable set consists of nodes with unique ranks, which by Observation~\ref{obs:structrank} guarantees that all the new sequences generated are mutually terminal-structure distinct. 

Second, the permutable set must be commonly located in the terminal structures of all sequences in the set. Thus, permutation evolution selects such a commonly located permutable set that is found in the terminal structures of a subset of the sequences, and removes the other sequences. This causes a decrease in distinctness. 

Third, in order to have a non-negligible effect on distinctness, the permutable set must be large enough. 
This is achieved by considering the number of roots, which depends on how the algorithm executed \opEm{} during the previous round. 
We classify the terminal structures of sequences based on the number of the roots, the position of the root with the largest subtree, as well as the number of its children (Section~\ref{ssub:big_small_evolution}). 
We pick the most common such root in all of the sequences and remove the rest of the sequences. This causes a decrease in distinctness. 
If the number of roots is small, this yields an adequate lower bound on the number of nodes in the largest subtree, and we perform permutation evolution.  
Otherwise, if the number of roots is large, we instead remove the root of the largest subtree by first performing a \opDc{} on that root so that it becomes the designated minimum root, and then calling \opEm{} (Section~\ref{ssub:designated_minimum_root_evolution}). 
This increases the number of roots further before an \opEm{} is called (Section~\ref{ssub:extract_min_evolution}). Intuitively, since \opEm{} has to at least pair all the roots and since we assumed an amortized upper bound of $\cem \log n$ on \opEm{}, we show that we cannot have a large number of roots in more than half of the rounds. 

Lastly, calling \opEm{} causes a number of issues. First, the algorithm can execute a different sequence of suboperations (each of different lengths) to implement \opEm{} on each sequence. 
We select the most common sequence of suboperations and remove all others. This ensures that the resulting set of sequences are still mutually algorithmically indistinguishable. It also causes a decrease in distinctness. 
In addition, \opEm{} can cause rank increases leading to marked nodes. 
By Lemma~\ref{lem:subdistinct}, the positions of these nodes in the terminal structures of all of the sequences are identical. Therefore, since by Lemma~\ref{lem:monotonerank} ranks never decrease, we can restore the invariant that rank of a node equals the negation of its key value by calling \opDc{} on each such node.

At the end, we show that the decreases in distinctness due to removing sequences are dominated by the increase in distinctness that occurs during the permutation evolution. We also show that the permutation evolution is invoked in at least half of the rounds (Lemma~\ref{lem:fracsmall}), implying that after sufficiently many rounds, the distinctness becomes $\omega(n)$, contradicting Lemma~\ref{lem:maxdistinct}, and thus showing that the only assumption we used in the proof, $\fdc = o(\log\log n)$, is false. 

Next, we present the evolutions in detail in Section~\ref{sub:evolutions} and show how to group them into rounds in Section~\ref{sub:rounds}. We analyze the running time of executing a sequence formed after $k$ rounds of evolutions in Section~\ref{sub:upper_bound_on_time}. Finally, we put it all together and prove the contradiction in Sections~\ref{sub:obtaining_a_contradiction}.


\subsection{Evolutions} 
\label{sub:evolutions}

We will now describe several functions on AI-TSD sets of heap operations; we call such functions \emph{evolutions}. The general idea is to append individual heap operations or small sequences of heap operations to all sequences in the input set $\Ids$ and remove some of the resulting sequences so as to maintain the property that the sequences in the resultant set of sequences $\Ids'$ are AI-TSD. The evolutions will also have the property that if the time to execute all sequences in $\Ids$ is identical, then the runtime to execute all sequences in $\Ids'$ will also be identical. The difference in the runtime to execute sequences in $\Ids'$ versus those in $\Ids$ will be called the \emph{runtime} of an evolution.

\subsubsection{Insert evolution} 
\label{ssub:insert_evolution}

The \emph{insert evolution} has the following form: $\Ids' = \eIns(\Ids)$.

In an insert evolution, a single \opIns\ operation of a key with value 0 is appended to the end of all $\Xi$ to obtain $\Xi'$. Given $\Ids$ is AI-TDS, the set $\Ids'$ is AI-TDS and trivially $\Distinct(\Ids)=\Distinct(\Ids')$. The runtime of the evolution is 1 since the added \opIns\ has runtime 1. The rank of the newly inserted node is 0, and is thus unmarked.


\subsubsection{Decrease-key evolution} 
\label{ssub:decrease_key_evolution}

The \emph{decrease-key evolution} has the following form: $\Ids' = \eDc(\Ids,p)$, where $p$ is a location which is either a root or a marked node in all terminal structures of sequences in $\Ids$. 

In a decrease key evolution, a \opDc$(p, \Delta x)$ operation is appended to the end of all sequences in  $\Xi$ to obtain $\Xi'$. 
The value of $\Delta x$ is chosen such that the new key value of what $p$ points to is set to the negation of its current rank; this means $\Delta x$ is always nonnegative because of the monotone property of ranks noted in Lemma~\ref{lem:monotonerank}. 
Observe that if $p$ points to a marked node, then it is unmarked after performing an \eDc. This requirement ensures that all structures that are distinct before this operation will remain distinct after the operation. Thus, the set $\Ids'$ is AI-TDS and trivially $\Distinct(\Ids)=\Distinct(\Ids')$. The runtime of the evolution is 1 since the added \opDc\ has runtime 1.


\subsubsection{Designated minimum root evolution} 
\label{ssub:designated_minimum_root_evolution}

Letting $\bigroot$ be the position of one root which exists in all terminal structures of $\Ids$, the  \emph{designated minimum root evolution} has the form $\Ids'=\eDmr(\Ids,r)$.

In a designated minimum root evolution, a \opDc\ operation on $\bigroot$ to a value of negative infinity is appended to all sequences in $\Ids$ to give $\Ids'$. It will always be the case that the (next) evolution performed on $\Ids'$ will be an extract-min evolution; the root $\bigroot$, which is known as the \emph{designated minimum root}, will be removed from all terminal structures of $\Ids'$ in this subsequent \eEm. There is no change in distinctness caused by this operation: $\Distinct(\Ids)=\Distinct(\Ids')$. The runtime of the evolution is 1 since the added \opDc\ has runtime 1.


\subsubsection{Extract-min evolution} 
\label{ssub:extract_min_evolution}

The \emph{extract-min evolution} has the form $(\Ids',\Vio,\acem)=\eEm(\Ids)$.

The extract-min evolution is more complex than those evolutions previously described, and the derivation of $\Ids'$ from $\Ids$ is done in several steps. 

First, an \opEm\ operation is appended to the end of all sequences in $\Ids$ to obtain an intermediate set of sequences which we call $\tent{\Ids}$. 
Recall that the \opEm\ operation is implemented by a number of suboperations.
There is no reason to assume that the suboperations executed by the algorithm in response to the \opEm\ in each of the elements of $\tent{\Ids}$ are the same; thus the set  $\tent{\Ids}$ may no longer be algorithmically indistinguishable. We fix this by removing selected sequences from the set $\tent{\Ids}$ so that the only ones that remain execute the appended \opEm\ by using identical sequences of suboperations. This is done by looking at the first suboperation executed in implementation of \opEm\ in each element of $\tent{\Ids}$, seeing which suboperation is the most common, and removing all those sequences  $\tent{\Ids}$ that do not use the most common first suboperation. If the suboperation is one which has a return value, the return value which is most common is selected and the remaining sequences are removed. This process is repeated for the second suboperation, etc., until the most common operation is \subop{End}{} and thus the end of all remaining suboperation sequences has been simultaneously reached. Since there are only a constant $\subops$ number of suboperations, and return values, if present, are boolean, at most a constant fraction (specifically $1-\frac{1}{2 \subops}$) of $\tent{\Ids}$ is removed while pruning each suboperation. At the end of processing each suboperation by pruning the number of sequences, the new set is returned as $\Ids'$. The set $\Ids'$ can be seen to be terminal-structure distinct, since pairing identically positioned roots in structurally different heaps, and having the same nodes win the pairings, can not make different structures the same. 

Observe that the nodes winning pairings in the execution of the $\opEm$ might have their ranks increase, and thus become marked. Moreover, due to Lemma~\ref{lem:subdistinct}, the position of all such nodes is identical in all terminal structures of $\Ids'$. The set of the locations of these newly marked  nodes is returned as $\Vio$, the \emph{violation set}.

Observe that it has been ensured that all of the sets of operations execute the appended \opEm\ using the same suboperations.  Thus, we 
define $\acem$ to be this common number of suboperations used to implement the \opEm; this value is returned by the evolution. As each suboperation reduces the distinctness by at most a constant,
$\Distinct(\Ids') \geq \Distinct(\Ids)-\acem \log ({2 \subops})=  \Distinct(\Ids)- \bigoh{\acem}$. The runtime of the evolution is $\acem$ since that is the cost of the added \opEm.


\subsubsection{Big/small evolution} 
\label{ssub:big_small_evolution}
The \emph{big/small evolution} has the form $(\Ids',(p,\text{\it{bigsmall}}))=\eBs(\Ids)$.

The goal of the big/small evolution is to ensure that the terminal structures of all sets are able to be executed in the same way in subsequent evolutions.
In a big/small evolution, each terminal structure of each of the operation sequences of $\Ids$ is classified according to the following, using the previously-defined function $\rcf$: 
\begin{itemize}

\item The exact number of roots (if at most $\rcf  \log n$) or the fact that the number of roots is greater than $\rcf  \log n$ (we call this case \emph{many-roots}). 

\item If the exact number of roots is at most ${\rcf  \log n}$:

\begin{itemize}

\item The position of the root with the largest subtree (the leftmost such root if there is a tie). Call it $p$. Observe that the size of $p$'s subtree is at least 
$\cst$.
\item The exact number of children of $p$ if less than ${\rcf  \log \cst}$ (we call this case \emph{small}) or the fact that the number of roots is greater than ${\rcf {} \log \cst}$ (we call this case \emph{root-with-many-children}).

\end{itemize}

\end{itemize}

There are at most $\lceil {\rcf {} \log n} \rceil^2 \cdot \lceil {\rcf {} \log \cst} \rceil$ possible classifications. We create set $\Ids'$ by removing from $\Ids$  sequences with all but the most common classification of their terminal structures. 
The return value is based on the resultant classification:
\begin{description}
\item[Many-roots:] Return $(p,\text{\it{bigsmall}})$ where $p=\text{\sc{Null}}$ and $\text{\it{bigsmall}}=\text{\sc{Big}}$.
\item[Root-with-many-children:] Return $(p,\text{\it{bigsmall}})$ where $p$ is the location of the root with the largest subtree and $\text{\it{bigsmall}}=\text{\sc{Big}}$.
\item[Small:] Return $(p,\text{\it{bigsmall}})$ where $p$ is the location of the root with the largest subtree and $\text{\it{bigsmall}}=\text{\sc{Small}}$.
\end{description}

We bound the loss of distinctness, which is the logarithm of the number of classifications.
 Since $\rcf=o(\log n)$, then $\log \left( \lceil {\rcf {} \log n} \rceil^2 \cdot \lceil {\rcf {} \log \cst} \rceil\right)= \bigoh{\log \log n}$, and thus $\Distinct(\Ids') = \Distinct(\Ids)- \bigoh{\log \log n}$. The evolution's runtime is 0 since no operations are added to any sequence.


\subsubsection{Permutation evolution} 
\label{ssub:permutation_evolution}

The \emph{permutation evolution} has the form $\Ids'=\ePerm (\Ids )$, where the leftmost root $\bigroot$ has in all terminal structures of the sequences of $\Ids$  a subtree size of at least $\frac{n  }{\rcf \log n}$ and at most ${\rcf  \log \cst}$ children; this will be achieved by being in the  small case of the big/small evolution and performing a decrease-key evolution on the relevant node.
It is required that all terminal structures of sequences in $\Ids$ are entirely unmarked.

In a permutation evolution, distinctness increases, and is the only evolution to increase the number of sequences in the process of converting $\Ids$ to $\Ids'$.

Combining the preconditions of the permutation evolution with Lemma~\ref{lem:effcount}, yields the fact that all nodes in the terminal structures of $\Ids$ at location $\bigroot$ have at least ${\rcg \log \cst}$ efficiently linked children; since there are at most $\ozwinpo \rce $ efficiently linked children of each rank (due to the definition of efficiently linked), that means there are at least  $\frac{\rcg \log \cst}{\ozwinpo\rce}$ efficiently linked children of different ranks in each terminal structure. Find such a set and call it the \emph{permutable set} (chose one arbitrarily if there is more than one possibility). Look at the position of these permutable sets in all terminal structures of the sequences of $\Ids$, and pick the position of the permutable set that is most common. Form the intermediate set of sequences $\hat{\Ids}$ by removing from $\Ids$ all sequences that do not have this commonly located permutable set.
Letting $F = f(n)\log\frac{n}{f(n)\log n}$ and $G=\frac{g(n)}{\ozwinpo \rce}\log\frac{n}{f(n)\log n}$,
an upper bound on the logarithm of the number of different locations a permutable set of size $G$ could be in is
\[
\log {F \choose G} = \Theta(G\log \frac{F}{G}) = \Theta\left(\frac{\log n}{\log\rcd}\cdot\log{\rcd}\right) = \Theta\left(\log n\right)
\]
given $f(n) = o(\log n)$ and $g(n) = \Theta(\rcd/\log\rcd)$. 
Thus this step decreases the distinctness of the set by at most the logarithm of the number of commonly located permutable sets:  
$$\Distinct(\hat{\Ids})-\Distinct(\Ids)=   - 
\bigoh{\log n}.$$

The permutable set is of size $\frac{\rcg \log \cst}{\ozwinpo\rce}$.  
Using the definitions of $\rcf$, $\rcg$, \ozwin, and $\fdc$, this is $\Theta\left(\frac{\log n}{ \log \fdc}\right)$. Let $\rcm$ be a constant such that the permutable set is of size at least $\frac{\rcm \log n}{ \log \fdc}$ for sufficiently large $n$. 
Recall that in all the terminal structures of the sequences of $\hat{\Ids}$, the ranks of the items in the permutable sets are different, and in fact are strictly increasing, when viewed right-to-left as children of $r$.

 We then create $\Ids'$ by replacing each sequence in ${\hat{\Ids}}$ with $\left(\frac{\rcm \log n}{ \log \fdc}\right)!$ new sequences created by appending onto the end of each existing sequence a sequence of all possible permutations of \opDc\ operations on all elements of an arbitrary subset of size $\frac{m \log n}{ \log \fdc}$ of the permutable set.

The fact that all of the sequences in ${\hat{\Ids}}$ have the same permutable sets ensures that all terminal structures in $\Ids'$ are terminal-structure distinct. 
 (Recall that Observation~\ref{obs:structrank} says that different ranks imply different structures of induced subtrees). 
Thus, in this step distinctness increases by $\Distinct(\Ids')-\Distinct({\hat{\Ids}})=\log \left( \frac{\rcm \log n}{ \log \fdc}  \right)! = \Theta\left( \frac{\log n \log \log n}{ \log \fdc} \right)$.

Thus, combining all the steps of the permutation evolution bounds the total increase of distinctness by
\begin{align*}
\Distinct(\Ids')-\Distinct(\Ids)&=
\Distinct(\Ids')-\Distinct(\hat{\Ids})+\Distinct(\hat{\Ids})-\Distinct(\Ids)\\
&=\Theta\left( \frac{\log n \log \log n}{ \log \fdc} \right) 
-\bigoh{\log n}\\
&=\Theta\left( \frac{\log n \log \log n}{ \log \fdc} \right) &&\text{since $\fdc = o(\log \log n)$}
\end{align*}

The cost of the evolution is at most $\frac{\rcm \log n}{ \log \fdc}$, since that is the number of \opDc\ operations appended to the sequences, and these all have unit cost.


\subsection{Rounds} 
\label{sub:rounds}


\begin{algorithm}
\caption{Algorithmic
 presentation of how evolutions are used to build the sequence of AI-TSD sequences 
 $\langle \Xi_0, \Xi_1, \Xi_2, \ldots \rangle$, which are split into rounds where the index of the start of round $i$ is $\circ_i$.
} 
\label{a:evolve}
\begin{algorithmic}
\State $\Ids_0=\{\langle \overbrace{\opIns(0), \opIns(0), \ldots, \opIns(0)}^{n\text{ \opIns\ operations}} \rangle \}$
\State $i=0$
\State $round=0$
\State $\round_0=0$;
\Loop
\State $(\Ids_{i},p,\text{\it{bigsmall}})= \eBs(\Ids_{i \pp})$; 
\If {$\text{\it{bigsmall}}=\text{\sc Small}$} 
\State $\Ids_{i} = \eDc(\Ids_{i \pp},p)$; \Comment{Small round}
\State $\Ids_{i}=\ePerm(\Ids_{i \pp})$; 
\Else 
\State $\Ids_{i}=\eDmr(\Ids_{i \pp},p)$; \Comment{Big round}
\EndIf
\State  $(\Ids_{i},\Vio,\acem)=\eEm(\Ids_{i \pp})$; \Comment{Common to small and big rounds}
\For{each $v$ in $\Vio$}
\State $\Ids_{i} = \eDc(\Ids_{i \pp},v)$;
\EndFor
\State $\Ids_{i} = \eIns(\Ids_{i \pp})$; 
\State $\round_{\pp round}=i$; 
\EndLoop
\end{algorithmic}

\end{algorithm}

We proceed to perform a sequence of evolutions $\Evo=\langle \evo_0, \evo_1, \ldots \rangle$ to define a sequence of AI-TSD sets $\langle \Ids_0, \Ids_1, \ldots \rangle$. The initial set $\Ids_0$ consists of a single sequence of operations: the operation \opIns$(0)$, executed $n$ times. Each subsequent AI-TSD  set $\Ids_{i}$ is derived from  $\Ids_{i-1}$  by performing the single evolution $\evo_{i-1}$; thus in general $\Ids_{i}$ is composed of some of the sequences of $\Ids_{i-1}$ with some operations appended. 

These evolutions are split into \emph{rounds}; $\round_i$ is the index of the first AI-TSD set of the $i$th round. Thus round $i$ begins with AI-TSD set $\Ids_{\round_i}$ and ends with $\Ids_{\round_{i+1}-1}$ through the use of evolutions
$\langle \evo_{\round_{i}} \ldots \evo_{\round_{i+1}-1}  \rangle$

 These rounds are constructed to maintain several invariants:
\begin{itemize}
\item
All terminal structures of all sequences in the AI-TSD set at the beginning and end of each round have size $n$. This holds as in each round, exactly one insert evolution and exactly one extract-min evolution is performed.
\item  All nodes in all terminal structures in the AI-TSD sets at the beginning and end of each round are unmarked.
\end{itemize}

There are two types of rounds, \emph{big rounds} and \emph{small rounds}. 
At the beginning of both types of rounds a big/small evolution is performed. The return value of the big/small evolution determines whether it will be a big or a small round.

The reader may refer to Algorithm~\ref{a:evolve} 
for a concise presentation of how evolutions are used to construct $\langle \Ids_0, \Ids_1, \Ids_2, \Ids_3, \ldots \rangle$. We now describe this process in detail.

\subsubsection{The Big Round} 
\label{ssub:the_big_round}

As the round begins, we know that the terminal structures of the AI-TSD set are entirely unmarked, and there are either at least ${\rcf  \log n}$ roots, or one root with at least ${\rcf {} \log \cst}$ children.  The round proceeds as follows:

\begin{enumerate}
\item Perform a designated minimum root evolution on the root with largest subtree; this is the node $\bigroot$ from the big/small evolution whose location is encoded in the return value; as a result of the big/small evolution it is guaranteed to be in the same location in all of the terminal structures of the sequences of $\Ids$. 
\item Perform an extract-min evolution. 
\item For each item in the violation sequence returned by the extract-min evolution, perform a decrease-key evolution. This makes the terminal structures of all heaps in $\Ids$ unmarked. 
\item Perform an insert evolution. 

\end{enumerate}

Assuming we are in round $i$, let $\acem_i$ be the cost of the extract-min evolution, and let $\vs_i$ be the size of the violation sequence.
The cost of the round (the sum of the costs of the evolutions) is $\acem_i + \vs_i+2$, which is at least ${\rcf {} \log \cst}$, and based on the evolutions performed the distinctness can be bounded as follows: 
$$\Distinct_i-\Distinct_{i+1}=\overbrace{\bigoh{\acem_i}}^\eEm +\overbrace{\bigoh{\log \log n}}^\eBs .$$


\subsubsection{The Small Round} 
\label{ssub:the_small_round}
In a small round, there is one root, call it $\akey$, at the same location in all terminal structures, with size at least $\frac{n  }{\rcf   \log n}$ and some identical number of children in all terminal structures which is at most ${\rcf {} \log \cst}$. The location of $\akey$ was returned by the big/small evolution. The round proceeds as follows:

\begin{enumerate}

\item Perform a decrease-key evolution on $\akey$ to make it negative infinity. 
\item  Perform an permutation evolution. 
\item Perform an extract-min evolution. 
\item For each item in the violation sequence returned by the extract-min evolution, perform a decrease-key evolution. 
\item Perform an insert evolution. 

\end{enumerate}

Let $\acem_i$ be the actual cost of the \opEm,  let $\vs_i$ be the size of the violation sequence.
The cost of the round is $\acem_i + \vs_i+2+\frac{\rcm \log n}{ \log \fdc}$, and based on the evolutions performed the distinctness can be bounded as follows: 
$$\Distinct_i-\Distinct_{i+1}=\overbrace{\bigoh{\acem_i}}^\eEm +\overbrace{\bigoh{\log \log n}}^\eBs-\overbrace{\varbigom{\frac{\log n \log \log n}{  \log \fdc}}}^\ePerm .$$



\subsection{Upper Bound on Time} 
\label{sub:upper_bound_on_time}

The following lemma is needed in the next section.

\begin{lemma}\label{lem:totaltime}
The total time to execute any sequence in $\Ids_{\round_k}$ is \varbigoh{\frac{\fdc}{\log \fdc}\cdot k\log n}.
\end{lemma}

\begin{proof}
Let $\opB$ be a sequence in $\Ids_{\round_k}$, and let $\runb$ be the time to execute $B$.
Let $\dcc_i$ be the cost of the permutation evolution in round $i$; this is at most $ \frac{\rcm \log n}{ \log \fdc}$ 
 if round $i$ is a small round and 0 if round $i$ is a big round (recall that permutation evolutions are only performed in the small round).
Thus the cost for any round $i$, whether big (\S\ref{ssub:the_big_round}) or small (\S\ref{ssub:the_small_round}), can be expressed as $2+\dcc_i +    \acem_i +  \vs_i$.
The time to execute any sequence $\opB \in \Ids_{\round_k}$, which we denote as \runb, is the sum of the costs of the rounds:

\begin{equation} \label{eq1}
 \runb =  \sum_{i=1}^k (2+\dcc_i +    \acem_i +  \vs_i) 
\end{equation}

An item in the violation sequence has had its rank increase. Its rank can only increase after winning $\rce$ pairings. Pairings are operations. Thus,

\begin{align}
 \sum_{i=1}^k \vs_i 
 &\leq \frac{\runb}{\rce }  \label{eq2}
\\
\intertext{Combining~\eqref{eq1} and~\eqref{eq2} and rearranging gives:}
\sum_{i=1}^k \vs_i &\leq \frac{1}{\rce}{\sum_{i=1}^k (2+\dcc_i +    \acem_i +  \vs_i)}
\nonumber
\\
\left(1-\frac{1}{\rce} \right)
\sum_{i=1}^k \vs_i &\leq \frac{1}{\rce}{\sum_{i=1}^k (2+\dcc_i +    \acem_i )}
\nonumber
\\
\sum_{i=1}^k \vs_i &\leq  \frac{1}{\rce -1}\sum_{i=1}^k (2+\dcc_i +   \acem_i )
\label{eq3}
\end{align}

We know by assumption that the runtime is bounded by the sum of the amortized costs stated at the beginning of \S\ref{sec:lower_bound}. Each round has one \opIns\ and one \opEm\ (at an amortized cost of $\cem \log n$ each) and $1+dc_i + \vs_i$ \opDc\ operations (at an amortized cost of $\fdc$ each). This gives:

\begin{equation} \label{eq4}
\runb\leq \sum_{i=1}^k (  (\dcc_i + \vs_i +1 ) {\fdc}  + 2 \cem \log n )
\end{equation}
Combining \eqref{eq3} and \eqref{eq4} gives:
\begin{align*}\label{eq5}
\runb & \leq \sum_{i=1}^k \left[ \left(  \dcc_i  +
\frac{2+\dcc_i +   \acem_i }{\rce -1} +1\right) \fdc
   + 2 \cem \log n \right]
\end{align*}
Since $\dcc_i \leq \frac{\rcm\log n}{\log \fdc}$ (recall that $\rcm$ is a constant defined in \S\ref{ssub:permutation_evolution}).

\begin{align*}
\runb &\leq \sum_{i=1}^k \left[ \left( 
\frac{2 }{\rce -1} +1 \right) \fdc +
 \frac{\acem_i \fdc}{\rce-1} 
   + 2 \cem\log n 
   + \left(1+\frac{1}{(\rce-1)} \right) \frac{m\fdc}{\log\fdc}\log n \right]
\\
\runb&\leq 
k \log n
\left[
\left(  \frac{2 }{\rce -1}+1 \right) \frac{\fdc}{\log n} 
   + 2\cem
   + \left( 1+\frac{1}{\rce-1}\right)  \frac{m\fdc}{\log\fdc} 
\right]
+
\frac{\fdc}{\rce-1}
\sum_{i=1}^k 
\acem_i 
\end{align*}

\noindent Since $\sum_{i=1}^k \acem_i \leq \runb$

\begin{align*}
\runb
\left(
1-\frac{\fdc}{\rce-1}
\right)
&\leq 
k \log n
\left[
\left(  \frac{2 }{\rce -1} +1\right) \frac{\fdc}{\log n} 
   + 2\cem
   + \left( 1+\frac{1}{\rce-1}\right)  \frac{m\fdc}{\log\fdc} \right]
\\
\intertext{Substituting in the definition of $\rce$: $\rce=2\fdc+1$}
\runb
\left(
1-\frac{1}{2}
\right)
&\leq 
k \log n
\left[
\left(  \frac{1 }{\fdc}+1 \right) \frac{\fdc}{\log n} 
   + 2\cem
   + \left( 1+\frac{1}{2\fdc}\right)  \frac{m\fdc}{\log\fdc} 
\right]
\\
\runb
&\leq 
2k \log n
\left[
\frac{1 }{\log n} 
+
\frac{\fdc}{\log n} +
2 \cem+
    \frac{m\fdc}{\log\fdc} 
    + \frac{m}{2\log\fdc} 
\right]
\end{align*}
Since $\fdc = o(\log\log n)$, for large enough $n$, 
\begin{equation} \label{eq:time}
\runb
\leq 
2k \log n
\left[
2 \cem+
    \frac{m\fdc}{\log\fdc} 
    + \frac{m}{2\log\fdc} 
\right]
< 
4k \log n
\left[
 \cem+
    \frac{m\fdc}{\log\fdc} 
\right].
\end{equation}
This yields $\runb = \varthet{\frac{\fdc}{\log\fdc}\cdot k \log n}.$

\end{proof}


\subsection{Obtaining a contradiction} 
\label{sub:obtaining_a_contradiction}

\begin{lemma}\label{lem:fracsmall}
After $k$ rounds, at least $\frac{k}{2}$ rounds must be small rounds.
\end{lemma}

\begin{proof}
Proof is by contradiction. Suppose more than $\frac{k}{2}$ rounds are big rounds.
The actual cost of a big round is at least ${\rcf {} \log \cst}$, so the total actual cost is greater than:

\begin{align*}
\runb & \geq \frac{k \rcf}{2} \log \cst \\
& = k \frac{\ftop}{4 \log \rce} \log \frac{2n \log \rce}{\ftop \log n}
& \text{since } \rcf=\frac{\ftop}{2\log \rce}
\intertext{since  $\rcd=2\fdc+1$ and $\fdc = o(\log \log n)$, for large enough $n$ we have}
& >  k \log n^{2/3} \cdot \frac{2^\ozpexp \fdc}{4 \log \fdc^{4/3}}  \\
& =  k\log n \cdot\frac{2^\ozpexp \fdc}{8 \log \fdc} \\
\intertext{Setting $\ozpexp = \ozexp$ yields}
& >  k\log n \left[\frac{32\cem \fdc}{8 \log \fdc} + \frac{32\rcm \fdc}{8 \log \fdc}\right]  \\
\runb & >  4k\log n \left[\cem  + \frac{\rcm\fdc}{\log \fdc}\right]  
\end{align*}
This contradicts Equation~\eqref{eq:time} in Lemma~\ref{lem:totaltime}. Thus, at least $\frac{k}{2}$ rounds must be small rounds.
\end{proof}

\begin{proof}[Proof of Theorem~\ref{thm:main}]
The distinctness gain of a round has been bounded as follows:
\begin{equation*}
 \Distinct_{\round_{i+1}}-\Distinct_{\round_i} =
\begin{cases}
-\bigoh{\acem_i}-\bigoh{\log \log n}  & \text{If the $i$th round is a big round (\S\ref{ssub:the_big_round})}\\
-\bigoh{\acem_i}-\bigoh{\log \log n} + \Theta \left( \frac{\log n \log \log n}{ \log \fdc} \right)
&\text{If the $i$th round is a small round (\S\ref{ssub:the_small_round})}
\end{cases}\label{eq5}
\end{equation*}
Now we know that $\sum_{i=1}^k \acem_i$ is less than the actual cost to execute a sequence in $\Ids_{\round_k}$, which is \varbigoh{\frac{\fdc}{\log \fdc}\cdot k\log n} by Lemma~\ref{lem:totaltime}. 
Substituting $\sum_{i=1}^k \acem_i = \varbigoh{\frac{\fdc}{\log \fdc}\cdot k\log n}$ into the above 
  and using the fact from Lemma~\ref{lem:fracsmall} that at least half of the rounds are small rounds gives:
\begin{equation} \label{eq6} 
\Distinct_{\round_k}-\Distinct_{\round_0} =  \varthet{ \frac{\log \log n}{ \log \fdc} \cdot k\log n}-\bigoh{k \log \log n}- \varbigoh{\frac{\fdc}{\log \fdc}\cdot k\log n}
\end{equation}
Since $\fdc=o(\log \log n)$, 
the negative terms in the previous equation can be absorbed, giving:
$$ \Distinct_{\round_k}-\Distinct_{\round_0} = \Theta \left( \frac{k\log n \log \log n}{ \log \fdc} \right) = \omega\left(k\cdot \frac{\log n\log\log n}{\log\log\log n}\right).
$$
But after sufficiently many rounds (i.e.~sufficiently large $k$) this contradicts Lemma~\ref{lem:maxdistinct} that for all $i$, $\Distinct_i = \bigoh{n}$. Thus for sufficiently large $k$ and $n$ a contradiction has been obtained, proving Theorem~\ref{thm:main}. 
\end{proof}




\bibliographystyle{abbrv}
\bibliography{dblpshort,bibshort}

\end{document}